\documentclass[12pt]{amsart}

\usepackage[english]{babel}

\usepackage{times,amsfonts,amsmath,amssymb,dsfont} 
\usepackage{pdfsync}
\usepackage{graphicx} 
\usepackage{mathabx}

\usepackage{cases}

\usepackage{soul}
\setstcolor{red}

\usepackage{color}

\definecolor{gr}{rgb}   {0.,   0.69,   0.23 }
\definecolor{bl}{rgb}   {0.,   0.5,   1. }
\definecolor{mg}{rgb}   {0.85,  0.,    0.85}
\definecolor{or}{rgb}   {0.9,  0.5,   0.}

\definecolor{webred}{rgb}{0.75,0,0}
\definecolor{webgreen}{rgb}{0,0.75,0}
\usepackage[citecolor=webgreen,colorlinks=true,linkcolor=webred]{hyperref}

\newtheorem{theorem}{Theorem}
\newtheorem{proposition}[theorem]{Proposition}
\newtheorem{lemma}[theorem]{Lemma}

\theoremstyle{definition}

\theoremstyle{remark}
\newtheorem{remark}[theorem]{Remark}

\newcommand{\ba}{\begin{array}}
\newcommand{\ea}{\end{array}}

\newcommand{\Bk}{\color{black}}

\newcommand{\N}{\mathbb{N}}
\newcommand{\R}{\mathbb{R}}
\newcommand{\C}{\mathbb{C}}

\newcommand{\cD}{\mathcal{D}}

\newcommand{\bel}{\begin{equation} \label}
\newcommand{\ee}{\end{equation}}

\begin{document}\setul{2.5ex}{.25ex}
\title[]{Resonances  near thresholds  in  slightly \Bk Twisted Waveguides}
\author[V.\ Bruneau]{Vincent Bruneau}
\address{Universit\'e de Bordeaux, IMB, UMR 5251, 33405 TALENCE cedex, France}\email{Vincent.Bruneau@u-bordeaux.fr}
\author[P.\ Miranda]{Pablo Miranda}\address{Departamento de Matem\'atica y Ciencia de la Computaci\'on, Universidad de Santiago de Chile, Las Sophoras 173. Santiago, Chile.}\email{pablo.miranda.r@usach.cl}
\author[N.\ Popoff]{Nicolas Popoff}
\address{Universit\'e de Bordeaux, IMB, UMR 5251, 33405 TALENCE cedex, France}
\email{Nicolas.Popoff@u-bordeaux.fr}
\email{}

\maketitle

\begin{abstract}
We consider the Dirichlet Laplacian in a straight three dimensional waveguide with non-rotationally invariant cross section, perturbed by a twisting of small amplitude. It is well known that such a perturbation does not create eigenvalues below the essential spectrum. However, around the bottom of the spectrum, we provide a meromorphic extension of the weighted  resolvent  of the perturbed operator,    and  show the existence of exactly one resonance near this point.  Moreover, we  obtain the  asymptotic behavior   of this resonance as the size of the twisting goes to 0.  We also extend the  analysis to the upper    eigenvalues of the transversal problem,  showing that the number of resonances is bounded by the multiplicity of the eigenvalue  and obtaining the corresponding asymptotic behavior. 
\end{abstract}

\section{Introduction}
Let $\omega$ be a bounded domain in $\R^2$  with Lipschitz boundary. 
Set $\Omega:=\omega\times\R$ and  $(x_1,x_2,x_3)=:(x_t,x_3)$.
Define $H_{0}$ as the Laplacian in $\Omega$ with Dirichlet boundary conditions. Consider    $-\Delta_\omega$ (the Laplacian in $\omega$ with Dirichlet boundary conditions). Since $\omega$ is bounded, the spectrum of the operator $-\Delta_\omega$
is a  discrete sequence of  values converging to infinity, denoted by $\{\lambda_n\}_{n=1}^\infty$. Then, the spectrum of $H_0$ is given by 
$$\sigma(H_0)=\bigcup_{n=1}^\infty[\lambda_n,\infty)=[\lambda_1,\infty),$$
and is purely absolutely continuous.

Geometric  deformations of such a straight waveguide have been widely studied in recent years, and have numerous applications in quantum transport in nanotubes. The spectrum of the Dirichlet Laplacian in 
waveguides provides   information about the quantum transport of spinless particles with {\it hardwall} boundary conditions. In particular, the existence of eigenvalues describes the occurrence of {\it bound states} corresponding to trapped trajectories created by the geometric deformations.  For a review 
we refer to \cite{Kre07}, where bending against twisting is discussed, and to \cite{Gru04} for  a general differential
approach.
Without being exhaustive we recall some well known situations: a local bending of the waveguide creates eigenvalues below the essential spectrum, as  also do a local enlarging of its width (\cite{DuEx95, Gru04}). On the contrary, it has been proved, under general assumptions, that a twisting of the waveguide does not lower the spectrum (\cite{EkKovKre08}), in particular a twisting going to 0 at infinity will not modify the spectrum (\cite{Gru04}). 
 In such a situation it is natural to introduce the notion of resonance and to analyze  the effect of the twisting on the resonances near the real axis. There already exist studies of resonances in waveguides: resonances in a thin curved waveguide (\cite{DuExMel97,Ned97}), or more recently in a straight waveguide with an electric potential, perturbed by a twisting (\cite{KovSacc07}). In these  both cases, however,  the resonances appear  as  perturbations of embedded eigenvalues of a reference operator, and follow the {\it Fermi golden rule} (see  \cite{Ha07} for references and for an overview on such resonances).  As we will see, in our case the origin of the resonances will be rather due to the presence of thresholds appearing as branch points created by a 1d Laplacian. Our analysis will be close to the studies near 0 of the 1d Laplacian (see for instance \cite{Si76,BulGesReSi97} where, even if resonances are not discussed,  the ``threshold" behavior appears). A similar phenomena of threshold resonances was already studied for a magnetic Hamiltonian  in   \cite{BonBruRai07},  where the   thresholds are eigenvalues of infinite multiplicity of some transversal problem.


In this article we will consider a small twisting of the waveguide: Let $\varepsilon: \R\to \R$ be a non-zero  function  of class $C^1$ \Bk with  exponential decay i.e.,  for some $\alpha >2(\lambda_2-\lambda_1)^{1/2}$ (this hypothesis can be relaxed, see Remark \ref{2oct17a}), $\varepsilon$ satisfies 
\bel{hypeps}
\varepsilon(x_3)=O(e^{-\alpha\langle x_3\rangle}), \quad   \varepsilon'(x_3)=O(e^{-\alpha\langle x_3\rangle}),
\ee
where $\langle x_3\rangle:=(1+x_3^2)^{1/2}$.
Then, for  $\delta>0$ we define $\Omega_{\delta}$ as the waveguide obtained by twisting $\Omega$ with  $ \theta_\delta$,   where $\theta_\delta'(x_3)=\delta\varepsilon(x_3)$, \Bk   \Bk i.e., we define 
 $$\Omega_{\delta}:=\{ (r_{\theta_\delta (x_{3})}(x_{t}),x_{3}), (x_{t},x_{3})\in \Omega)\},$$
 where $r_{\theta}$ is the rotation of angle $\theta$ in $\R^{2}$. Set $$W(\delta):=-\delta\partial_\varphi\varepsilon\partial_3-\delta\partial_3\varepsilon\partial_\varphi-\delta^2\varepsilon^2\partial_\varphi^2=-2\delta\varepsilon\partial_\varphi\partial_3-\delta\varepsilon'\partial_\varphi-\delta^2\varepsilon^2\partial^2_\varphi,$$
with the notation $\partial_\varphi$ for $x_1\partial_2-x_2\partial_1$. Then, it is standard (see for instance \Bk \cite[Section 2]{Gru04}) 
that the Dirichlet Laplacian in $\Omega_{\delta}$ is unitarily equivalent to the operator 
$$H(\delta):=H_{0}+W(\delta),$$
defined in $\Omega$ with  a  Dirichlet boundary condition.  Since the perturbation is a  second order \Bk differential operator, $H(\delta)$ is not a relatively compact perturbation of $H_{0}$. However the resolvent difference $H(\delta)^{-1}-H_{0}^{-1}$ is compact (\cite[Section 4.1]{BriKoRaSo09}), and therefore   $H(\delta)$ and $H_{0}$  \Bk have the same essential spectrum. 
Moreover, 
the spectrum of $H(\delta)$ coincide with  $[\lambda_{1},+\infty)$, see  \cite{EkKovKre08}. 

In this article we will show that  around 
 $\lambda_{1}$ there exists,  for $\delta$ small enough, a meromorphic extension of the  weighted  resolvent  of $H(\delta )$ with respect to the variable  $k:= \sqrt{z-\lambda_1}$,  \Bk where $z$ is the spectral parameter.  In other words,  the resolvent   $(H(\delta)- z)^{-1}$, first defined  for $z$ in    $\C\setminus [0, + \infty)$, admits a meromorphic extension  on  a weighted space (space of functions with exponential decay along the tube),  for values in   a neighborhood of $\lambda_1$ in   a 2-sheeted Riemann surface.  We will identify  the resonances \Bk
 around $\lambda_1$ with the poles of this meromorphic extension in the  parameter $k$. We will prove  in Theorem \ref{T:main} that in a neighborhood independent of $\delta$,  there is exactly one pole $k(\delta)$, whose behavior as $\delta\to0$ is explicit: 
 \bel{E:MT}
 k(\delta)=-i\mu\delta^2+O(\delta^3),
 \ee
 where $\mu>0$ is given by \eqref{E:Defmu}  below,  and moreover, $k(\delta)$ is on the imaginary axis.
  
The fact that $k(\delta)$ is on the negative imaginary axis means that in the spectral variable the resonance is on the second sheet of  the  2-sheeted Riemann surface, far from the real axis (it is sometimes called an antibound state \cite{Sim2000}). In particular such  a  resonance can not be detected using dilations (a dilation of angle larger than  $\pi$ would be needed) \Bk and is completely different in nature from those created by perturbations of embedded eigenvalues. For this reason  we define resonances as the poles of weighted resolvents, assuming  that $\varepsilon$ is exponentially decaying. 
However, \Bk  a difficulty  comes from the non relatively compactness of the perturbation $W(\delta).$ 
This problem   will be  
overcome exploiting the smallness of the perturbation   and the  locality  of our problem.

Our analysis provides   an analogous result for  higher thresholds,   in Section \ref{R:higher}:  Around  each $\lambda_{q_{0}}$ there are at most $m_{0}$ resonances (for all   $\delta$  small enough), where $m_{0}$ is the multiplicity of $\lambda_{q_{0}}$ as 
eigenvalue of $-\Delta_{\omega}$. Moreover,  
under an additional assumption, each one of these resonances have an asymptotic behavior  of the form \eqref{E:MT}, \Bk 
 where the constant $\mu$ 
is an   eigenvalue of a $m_0\times m_0$ explicit matrix (not necessarily Hermitian). 
 Although Theorem \ref{T:higher}  may be viewed as a generalization of Theorem \ref{T:main}, we preferred to push forward the proof for the 
first threshold  for the following reasons: it is easier to follow  and contain all the main ingredients needed for the proof in the upper  thresholds,   the  eigenvalues of $-\Delta_{\omega}$ are generically simple  as we know  the first eigenvalue is. 

  \begin{remark}
 Independent  of the size of the  perturbation $W(\delta)$, a more global definition of resonances would be possible by showing that a generalized determinant (as in \cite{BouBru08} or in \cite[Definition 4.3]{Sjo14}) is well defined on $\C\setminus [0, + \infty)$ and admits an analytic extension. Then the resonances would be defined as the zeros of this determinant  on a infinite-sheeted Riemann surface (as in \cite[Definitions 1-2]{BonBruRai07}). 
 \end{remark}

\section{Preliminary decomposition of the  free  resolvent}\label{S1}

 Let us describe the singularities of the free resolvent. Setting  $D_3:=-i\partial_3$, we have that  
\bel{24aug17aa}
H_{0}-\lambda_{1}=(-\Delta_{\omega}-\lambda_{1})\otimes I_{x_3}+I_{x_t}\otimes D_{3}^2.
\ee
For $k\in \C^{+}:= \{k \in \C; \;  {\rm Im \,} k >0 \}$, \Bk define
$$R_{0}(k):=(H_{0}-\lambda_{1}-k^2)^{-1},$$
and $R$ similarly for $H(\delta)$.  If   for $n\in\N$,  $\pi_n$
is  the orthogonal projection onto ker$(-\Delta_\omega-\lambda_n)$, using 
\eqref{24aug17aa}   for  $ k^2 \in \C\setminus [0, + \infty)$, \Bk we have
 that \Bk
\bel{24aug17a}R_0(k)=(H_0-\lambda_1-k^2)^{-1}=\sum_{q=1} \pi_q\otimes(D^2_{3}+(\lambda_q-\lambda_1)-k^2)^{-1}.\ee
The integral kernel of $(D_3^2-k^2)^{-1}$ is explicitly given 
by 
\bel{3sep17a}
\frac{i}{2k}e^{i\,k|x_3-x_3'|}.
\ee

Let $\eta$ be an exponential weight of the form  $\eta(x_{3})=e^{-N \left\langle x_{3}\right \rangle}$, for $(\lambda_2-\lambda_1)^{1/2}<N<\alpha/2$.
Also, for $a\in\C$ and $r>0$  set  $B(a,r):=\{z\in\C;|a-z|<r\}$.
Then, as in  \cite[Lemma 1]{BonBruRai07} it can be seen that  the operator valued-function $k\mapsto (R_{0}(k):\eta^{-1}L^2(\Omega)\to\eta L^2(\Omega)) $,  initially defined on $\C^+$, has a meromorphic extension \Bk in $B(0,r)$ for any $0<r<(\lambda_2-\lambda_1)^{1/2}$, with a unique pole,  of multiplicity  one, at $k=0$. \Bk  More precisely, 
\bel{6sep17a}
\eta R_{0}(k)\eta = \frac{1}{k} \pi_1 \otimes \alpha_0 + A_0(k),
\ee
where $\alpha_0$ is the rank one operator $\alpha_0 = \frac{i}{2} |\eta \rangle\langle \eta|$ and 
  $k\mapsto(A_0(k)$: $L^2(\Omega)\to L^2(\Omega))$ is the analytic operator-valued function 
 \bel{7sep17a}
 A_0(k) := \pi_1 \otimes r_1(k) + \sum_{q=2} \pi_q\otimes \eta (D^2_{3}+(\lambda_q-\lambda_1)-k^2)^{-1}\eta,
 \ee
  with $r_1$ being the operator in $L^2(\R)$ with integral kernel given 
by 
$$i\eta(x_3) \frac{ (e^{i\,k|x_3-x_3'|}-1)}{2k} \eta(x_3').$$
Clearly, for $0<r<(\lambda_2-\lambda_1)^{1/2}$, the family of operators $A_0(k)$ is uniformly bounded on $B(0,r)$.

\begin{remark}\label{2oct17a}
Note that the condition $\alpha >2(\lambda_2-\lambda_1)^{\frac12}$ on the function $\varepsilon$, enters here in order to have analytic properties in the ball $B(0,r)$,  $0<r<(\lambda_2-\lambda_1)^{1/2}$. This assumption can be relaxed to $\alpha >0$, but  the results will be restricted to $B(0,r)$ with  $0<r< \frac{\alpha}{2}$.
\end{remark}

In order to define and study the resonances, we will consider a suitable meromorphic extension of $R(k)$, using the identity 
\bel{4sept17e}
\eta R(k)\eta=\eta R_{0}(k)\eta\left({\rm Id}+\eta^{-1} W(\delta)R_{0}(k)\eta\right)^{-1}.\ee
Since $H(\delta)$ has no eigenvalue below $\lambda_1$ (see \cite{EkKovKre08}), the above relation is initially well defined and analytic for $k \in \C^+$. \Bk
 It is  necessary then to understand under which conditions this formula can be used to define  such an  extension. 
 Since  we can not apply directly the meromorphic   Fredholm theory  ($W(\delta)$ is not $H_{0}$-compact), we will need to show explicitly  that $\left({\rm Id}+\eta^{-1} W(\delta)R_{0}(k)\eta\right)^{-1}$ is meromorphic in some region around zero. 
 
 Let $\psi_1$ be such that $-\Delta_\omega \psi_1=\lambda_1 \psi_1$, $\|\psi_1\|_{L^2(\omega)}=1$ (then     $\pi_1=| \psi_1 \rangle \langle \psi_1 |$),   and define 
\bel{E:defPhi}\Phi_\delta:= -\frac{i}{2}((\partial_\varphi  \psi_1 \otimes \eta^{-1}  \varepsilon' )
+
\delta(\partial^2_\varphi \psi_1 \otimes  \eta^{-1} \varepsilon^2 )).
\ee

\begin{lemma}\label{Lholo0}
Let $0<r<(\lambda_2-\lambda_1)^{1/2}$.    There exists $\delta_0>0$ such that 
for any $0<\delta\leq\delta_0$ and $k\in B(0,r)\setminus\{0\}$  
$$\eta^{-1} W(\delta)R_{0}(k)\eta=\frac{\delta}{k} K_0 + \delta T(\delta,k),$$
where  $K_0$ is the rank one operator
\bel{4sep17} K_0:= |\Phi_\delta\rangle\langle \psi_1 \otimes \eta|,
\ee
and 
$B(0,r)\ni \Bk k\mapsto(T(\delta,k)$: $L^2(\Omega)\to L^2(\Omega))$ is an analytic operator-valued function. Moreover,  
\bel{2sep17}\sup_{0<\delta\leq\delta_0, \, k \in B(0,r)}||T(\delta,k)||<\infty.\ee

\end{lemma}
\begin{proof}
Thanks to \eqref{6sep17a}, 
\bel{E:devWR}
\eta^{-1} W(\delta)R_{0}(k)\eta= \frac{1}{k} \eta^{-1} W(\delta) \eta^{-1} ( \pi_1 \otimes \alpha_0) + \eta^{-1} W(\delta)\eta^{-1}  A_0(k). 
\ee
Since the range of the operator $\eta^{-1}  \alpha_0=\frac{i}{2}|1\rangle \langle \eta$ is spanned by constant functions, we have $\partial_3 \eta^{-1}  \alpha_0=0$, and therefore 
$$\eta^{-1}  W(\delta) \eta^{-1}( \pi_1 \otimes \alpha_0)=\frac{i}{2} |  \eta^{-1}  (-\delta\varepsilon'\partial_\varphi-\delta^2\varepsilon^2\partial^2_\varphi) \eta^{-1} (\psi_1 \otimes \eta)\rangle\langle \psi_1 \otimes \eta| = \delta |\Phi_\delta\rangle\langle \psi_1 \otimes \eta|=\delta K_{0}.$$
We now treat the last term of \eqref{E:devWR}: Setting  $ \delta T(\delta,k)= \eta^{-1} W(\delta)\eta^{-1}  A_0(k)$ we immediately get
\begin{align*}T(\delta,k)
=&-2
\Big(\partial_\varphi \pi_1 \otimes  \eta^{-1} \varepsilon \partial_3  \eta^{-1} r_1(k) + \sum_{ q\geq 2  } \partial_\varphi \pi_q\otimes  \eta^{-1} \varepsilon \partial_3 (D^2_{3}+(\lambda_q-\lambda_1)-k^2)^{-1}\eta\Big)\\
&- \Big(\partial_\varphi \pi_1 \otimes  \eta^{-1}  \varepsilon'   \eta^{-1} r_1(k) + \sum_{q\geq 2} \partial_\varphi \pi_q\otimes  \eta^{-1}  \varepsilon' (D^2_{3}+(\lambda_q-\lambda_1)-k^2)^{-1}\eta \Big)\\&
- \delta \Big( \partial_\varphi^2 \pi_1 \otimes  \eta^{-1}  \varepsilon^2  \eta^{-1}  r_1(k) + \sum_{q\geq 2}  \partial_\varphi^2 \pi_q\otimes   \eta^{-1}  \varepsilon^2 (D^2_{3}+(\lambda_q-\lambda_1)-k^2)^{-1}\eta\Big).\end{align*}
It is clear  that the two last terms are analytic and uniformly bounded in $B(0,r)$. For the first one, we note that the kernel of $\partial_{3}\eta^{-1}r_{1}(k)$ is $(x_{3},x_{3}') \mapsto -\tfrac{1}{2}\eta(x_{3}')\mathrm{sign} (x_{3}-x_{3}')e^{ik|x_{3}-x_{3}'|}$, and therefore $\partial_\varphi \pi_1 \otimes  \eta^{-1} \varepsilon \partial_3  \eta^{-1} r_1$  admits an analytic expansion which is uniformly bounded. The same arguments run for $\sum_{q\geq 2} \partial_\varphi \pi_q\otimes  \eta^{-1} \varepsilon \partial_3 (D^2_{3}+(\lambda_q-\lambda_1)-k^2)^{-1}\eta$.
\end{proof}

\section{Meromorphic extension of the resolvent and study of the resonance}\label{S3}

\begin{proposition}\label{propinv}
Let $\mathcal{D}\subset B(0,\sqrt{\lambda_2-\lambda_1})$ be  a compact  neighborhood  of zero. With the notation of Lemma \ref{Lholo0}, for $\delta$ sufficiently small, let us introduce the functions $\tilde{\Phi}_{\delta}=({\rm Id}+ \delta T(\delta,k) )^{-1}  \Phi_\delta$ and
\bel{4sept17c}
w_{\delta}(k)= \delta \langle \tilde{\Phi}_{\delta}| \psi_1 \otimes \eta \rangle.
\ee
Then: 
\begin{enumerate}
\item[i)] There exists $\delta_0$ such that for any $k \in \mathcal{D}$, $\delta \in (0,\delta_0)$, 
\bel{E:expw}
 w_\delta(k)= i\mu \delta^2+O(\delta^3)+\delta^2 k g_\delta(k)
\ee
where 
\bel{E:Defmu}
\mu:=\tfrac{1}{2}\sum_{q\geq2}(\lambda_{q}-\lambda_{1})\langle \partial_{\varphi}\psi_{1}|\pi_{q}\partial_{\varphi}\psi_{1}\rangle\langle \varepsilon|(D_{3}^2+\lambda_{q}-\lambda_{1})^{-1}\varepsilon\rangle
\ee  is  a  positive  constant,  and 
$g_{\delta}$ is an analytic function in $\mathcal{D}$   satisfying 
$$\sup_{\delta \in (0,\delta_{0})}\sup_{k\in\cD}|g_{\delta}(k)|<+\infty.$$
\item[ii)] When $\alpha\in \R$, there holds $w_{\delta}(i\alpha)\in i\R$.
\end{enumerate}

\end{proposition}
\begin{proof}
We use the Taylor expansion and Lemma \ref{Lholo0} to see that 
\bel{12oct17}({\rm Id}+\delta T(\delta,k))^{-1} ={\rm Id}-\delta T(\delta, 0 )+\delta k G_{\delta}(k) +O(\delta^2),\ee where $G_{\delta}(k)$ is holomorphic operator-valued function that is  uniformly bounded for  $ k\in\mathcal{D}$ and $\delta$ small.  

By definition of $\Phi_\delta$, we have:
$$ \langle \Phi_\delta | \psi_1 \otimes \eta\rangle 
= -\tfrac{i}{2}\left( \langle \partial_\varphi \psi_1 | \psi_1\rangle_{L^2(\omega)} \, \langle \eta^{-1} \varepsilon' | \eta \rangle_{L^2(\R)}  +  \delta \langle \partial^2_\varphi\psi_1 | \psi_1 \rangle_{L^2(\omega)} \, \langle  \eta^{-1} \varepsilon^2 | \eta \rangle_{L^2(\R)}\right).
$$
The first term is zero because $\varepsilon$  tends to zero \Bk at infinity. 
Using  integration by part, since $\psi_1$ satisfies a Dirichlet boundary condition, we deduce 
$$\langle \Phi_\delta | \psi_1 \otimes \eta\rangle =  \delta\frac{i}{2} \|\partial_{\varphi}\psi_{1}\|^2\| \varepsilon \|^2.$$
Noticing that $\| \Phi_{\delta} \|=O(1)$, from \eqref{12oct17} we get
\bel{E:expandwstart}
w_\delta(k)=\delta^2\frac{i}{2}\|\partial_{\varphi}\psi_{1}\|^2\| \varepsilon \|^2-\delta^2 \langle T(\delta,0 )  \Phi_\delta | \psi_{1} \otimes \eta \rangle+  \delta^2 k g_\delta(k)  + O(\delta^3), \Bk
\ee
where  $g_\delta(k)$ is holomorphic and  uniformly bounded for   $ k\in\mathcal{D}$ and $\delta$ small. 

 We now compute $\langle T(\delta,0)  \Phi_\delta | \psi_{1} \otimes \eta \rangle$. 
First recall  that   $ T(\delta,k)=\delta^{-1} \eta^{-1} W(\delta)\eta^{-1}  A_0(k)$. 
Next,  note that since $ \langle \partial_\varphi \psi_1 | \psi_1\rangle=0$,  
$$\pi_{1}\partial_{\varphi}\psi_{1}=0,$$
and therefore,
 using the definition of $\Phi_\delta$ in  \eqref{E:defPhi}, we get
$$(\pi_{1}\otimes r_{1}(0))\Phi_{\delta}=-\delta \frac{i}{2}  \pi_{1} \partial^2_\varphi \psi_1 \otimes r_{1}(0)  \eta^{-1} \varepsilon^2, $$
which in turn implies that
$$\langle (\delta^{-1}\eta^{-1}W(\delta) \eta^{-1}) (\pi_{1}\otimes r_{1}(0))\Phi_{\delta}  |  \psi_{1}\otimes \eta \rangle=O(\delta).$$
In consequence, having in mind  \eqref{E:defPhi} again, 
we deduce 
\begin{equation}
\label{E:expandTscalar}
\langle T(\delta,0 ) \Phi_{\delta},\psi_{1}\otimes \eta\rangle
\end{equation}
\begin{align*}=&
\langle \eta^{-1} \left(-2\varepsilon\partial_{\varphi}\partial_{3}-\varepsilon' \partial_{\varphi}-\delta \varepsilon^2\partial_{\varphi}^2\right) \eta^{-1}   \big(\sum_{q \geq 2} \pi_q\otimes \eta (D^2_{3}+(\lambda_q-\lambda_1))^{-1} \eta\big)\Phi_{\delta}  |  \psi_{1}\otimes \eta \rangle+O(\delta)
\\
=&\frac{i}{2}\langle \eta^{-1} \big(2\varepsilon\partial_{\varphi}\partial_{3}+\varepsilon' \partial_{\varphi}\big)     \big(\sum_{q \geq 2} \pi_q\otimes  (D^2_{3}+(\lambda_q-\lambda_1))^{-1}\big)\partial_\varphi  \psi_1 \otimes \varepsilon'  | \psi_{1}\otimes \eta \rangle+O(\delta).
\end{align*}

We compute the main term of the last expression using   integration by parts, both in the $\varphi$ and the $x_{3}$ variables:  
\begin{align}
\label{E:interTdelta}\begin{split}&
\langle \eta^{-1} \left(2\varepsilon\partial_{\varphi}\partial_{3}+\varepsilon' \partial_{\varphi}\right)\left(\sum_{q \geq 2} \pi_q\otimes  (D^2_{3}+(\lambda_q-\lambda_1))^{-1}\right)\partial_\varphi  \psi_1 \otimes \varepsilon' | \psi_{1}\otimes \eta \rangle
\\=&\sum_{q\geq2}\langle \partial_{\varphi}\psi_{1}|\pi_{q}\partial_{\varphi}\psi_{1}\rangle \times \langle \varepsilon'|(D_{3}^2+\lambda_{q}-\lambda_{1})^{-1}\varepsilon'\rangle.
\end{split}\end{align}
Now, we notice that 
\begin{align*}
\langle \varepsilon'|(D_{3}^2+\lambda_{q}-\lambda_{1})^{-1}\varepsilon'\rangle&=\langle \varepsilon |(D_{3}^2+\lambda_{q}-\lambda_{1})^{-1}D_{3}^2\varepsilon\rangle 
\\&= \|\varepsilon\|^2-(\lambda_{q}-\lambda_{1})\langle \varepsilon|(D_{3}^2+\lambda_{q}-\lambda_{1})^{-1}\varepsilon\rangle.
\end{align*}
In addition,   since $\pi_{1}\partial_{\varphi}\psi_{1}=0$ and $\sum_{q\geq1}\pi_{q}=\mathrm{Id}$,  we have that 
\bel{E:sumpositiv}
\sum_{q\geq 2}\langle \partial_{\varphi}\psi_{1}|\pi_{q}\partial_{\varphi}\psi_{1}\rangle =\|\partial_{\varphi}\psi_{1}\|^2.
\ee
Then,  from \eqref{E:expandTscalar} and \eqref{E:interTdelta} we get 
\bel{26sep17}\ba{ll}
&\langle T(\delta,0) \Phi_{\delta},\psi_{1}\otimes \eta\rangle
\\[.5em]
=&\displaystyle{\tfrac{i}{2}\|\varepsilon\|^2\|\partial_{\varphi}\psi_{1}\|^2-\tfrac{i}{2}\sum_{q\geq2}(\lambda_{q}-\lambda_{1})\langle \partial_{\varphi}\psi_{1}|\pi_{q}\partial_{\varphi}\psi_{1}\rangle\langle \varepsilon|(D_{3}^2+\lambda_{q}-\lambda_{1})^{-1}\varepsilon\rangle +O(\delta)}.
\ea\ee
Putting together  \eqref{E:expandwstart} and \eqref{26sep17},  we deduce \eqref{E:expw}. 
Moreover, $\mu$ is clearly non-negative, and from \eqref{E:sumpositiv}, there exists $q \geq 2$ such that $\langle \partial_{\varphi}\psi_{1}|\pi_{q}\partial_{\varphi}\psi_{1}\rangle>0$. Since $(D_{3}^2+\lambda_{q}-\lambda_{1})^{-1}$ is a positive operator, we get $\mu>0$. 

 Let us prove  ii.  For all $\alpha\in \R$, $A_{0}(i\alpha)$ has a real  integral kernel, see \eqref{7sep17a}. Therefore if $u\in L^{2}(\Omega)$ is real valued, so is $({\rm Id}+\delta T(\delta,i\alpha))^{-1}u$.  In consequence,  since $\Phi_{\delta}$ has values in $i\R$, so is $\tilde{\Phi}_{\delta}=({\rm Id}+\delta T(\delta,i\alpha))^{-1}\Phi_{\delta}$, and we deduce that $w_{\delta}(i\alpha)$ has  values in $i\R$ as well.
\Bk
\end{proof}

\begin{theorem}
\label{T:main}
Let   $\varepsilon: \R\to \R$ be a  non-zero $C^1$-function  satisfying \eqref{hypeps} and \Bk $\mathcal{D}\subset B(0,\sqrt{\lambda_2-\lambda_1})$ be a compact neighborhood of zero. Then, for $\delta$ sufficiently small, $ k\mapsto R(k)=(H-\lambda_{1}-k^2)^{-1}$,  initially defined in $\C^{+}$, admits a meromorphic operator-valued extension  on   $\cD$, whose operator-values act from $\eta^{-1}L^{2}(\Omega)$ into $\eta L^{2}(\Omega)$. \Bk This function has  exactly one  pole $k(\delta)$ in $\cD$, called a resonance of $H$, and it is of multiplicity one. Moreover, we have the asymptotic expansion
$$k(\delta)=-i\mu\delta^2+O(\delta^3),$$
 with $\mu$ given by \eqref{E:Defmu}  and $\mathrm{Re}(k(\delta))=0$.
\end{theorem}
\begin{proof}
Consider   the identity \eqref{4sept17e}, and note that from  Lemma \ref{Lholo0} 
for  $k\in \mathcal{D}\setminus\{0\}$ \Bk and 
\Bk \ $\delta$  sufficiently small we can  write  
\bel{27sep17}\Big( {\rm Id} + \eta^{-1} W(\delta)R_{0}(k)\eta \Big)=
\Big( {\rm Id} +\delta T(\delta,k) \Big)
\Big( {\rm Id} +\frac{\delta}{k}({\rm Id}+ \delta T(\delta,k) )^{-1} K_0  \Big).\ee
For $k\in \cD\setminus\{0\}$ let us set
$$K:=\frac{\delta}{k}({\rm Id}+ \delta T(\delta,k) )^{-1} K_0=\frac{\delta}{k} |\tilde{\Phi}_{\delta}><\psi_{1}\otimes \eta|,$$ 
which is a rank one operator. Then, we need to study the inverse  of  $({\rm Id}+K)$. 

Let us  consider $\Pi_{\delta}^{\perp}$, the projection onto $(\mbox{span }\{\psi_1 \otimes \eta\})^{\perp}$ into the direction $\tilde{\Phi}_\delta$ and $\Pi_{\delta}={\rm Id}-\Pi_{\delta}^{\perp}$, the projection onto $\mbox{span }\{\tilde{\Phi}_\delta\}$ into the direction normal to $(\psi_1 \otimes \eta)$. 
\Bk We can easily see that \Bk $$({\rm Id}+K) \Pi_{\delta}^{\perp} =\Pi_{\delta}^{\perp}  \qquad \mbox{and} \qquad ({\rm Id}+K) \Pi_{ \delta} =(1+\frac{\delta}{k}\langle \tilde{\Phi}_{\delta} | \psi_{1}\otimes \eta \rangle) \Pi_{ \delta}=\frac{k+w_{\delta}(k)}{k}\Pi_{\delta}.$$
Therefore, ${\rm Id}+K$ is invertible if and only if $k+w_{\delta}(k)\neq 0$, and
\bel{29sept17a}
({\rm Id}+K)^{-1}=\Pi_{\delta}^{\perp}+\frac{k}{k+w_{\delta}(k)} \Pi_{\delta}.
\ee
 Let us consider the equation $k+w_{\delta}(k)=0$. Using \eqref{E:expw}, for all $\kappa\in (0,\sqrt{\lambda_{2}-\lambda_{1}})$, for $\delta$ small enough, the equation has no solution for $k\in \cD$ and $|k|\geq \kappa$. We then apply Rouch\'e Theorem inside the ball $B(0,\kappa)$: consider the analytic  functions $h_{\delta}(k)=i\mu\delta^2 +  k$  and $f_{\delta}(k)=w_{\delta}(k)+k$. The function $h_{\delta}$ has exactly one root, and on the circle $C(0,\kappa)$, using again \eqref{E:expw},  there holds $|h_{\delta}-f_{\delta}| \leq h_{\delta}$ for $\delta$ small enough. 
\Bk Thus, we deduce that the equation $k+w_{\delta}(k)=0$ has exactly one solution $k(\delta)$ in $\cD$, for each fixed $\delta$ small enough. 
In consequence, putting together  \eqref{4sept17e}, \eqref{6sep17a}, \eqref{27sep17} and \eqref{29sept17a}  we have that  for all $k\in \cD\setminus \{ 0, k(\delta)\}$
$$\eta R(k)\eta = \Big( \frac{1}{k} \pi_1 \otimes \alpha_0 + A_0(k) \Big)
\Big( \Pi_{\delta}^{\perp} + \frac{k}{k+w_\delta(k)} \Pi_{\delta} \Big) \; ({\rm Id}+ \delta T(\delta,k) )^{-1}.$$

By the  definition of $\Pi_{\delta}^{\perp}$, 
we have that $ (\pi_1 \otimes \alpha_0) \Pi_{\delta}^{\perp}=0$ and then:
  $$\eta R(k)\eta = \frac{1}{k+w_\delta(k)} (\pi_1 \otimes \alpha_0)\Pi_{\delta} ({\rm Id}+ \delta T(\delta,k) )^{-1} $$
  $$+ \frac{k}{k+w_\delta(k)} A_0(k)\Pi_{\delta}  ({\rm Id}+ \delta T(\delta,k) )^{-1} +  A_0(k)  \Pi_{\delta}^{\perp}  ({\rm Id}+ \delta T(\delta,k) )^{-1}.$$
  Therefore, for $\delta$ sufficiently small, $k \mapsto \eta R(k)\eta $ admits a meromorphic extension to $\cD\setminus\{k(\delta)\}$, where the pole $k(\delta)$ is giving by the solution of   
  $k+w_\delta(k)=0$. 
  
 Using \eqref{E:expw}, the asymptotic expansion of $k(\delta)$ follows immediately. \Bk 
Further,  the multiplicity of this resonance is the rank of the residue of $\eta R(k)\eta$, which coincides with the rank of
 $(\pi_1 \otimes \alpha_0)\Pi_{\delta} + k(\delta)A_0(k(\delta)) \Pi_{\delta}$. It is one because $\Pi_{\delta}$ is of rank one with its range in span$\{\tilde{\Phi}_\delta\}$ and 
 $$\Big( (\pi_1 \otimes \alpha_0) + k(\delta)A_0(k(\delta))\Big)\tilde{\Phi}_\delta= \frac{i}{2} \langle \tilde{\Phi}_{\delta}| \psi_1 \otimes \eta \rangle (\psi_1 \otimes \eta)  + O(\delta^2)= - \frac{\delta\mu}2  (\psi_1 \otimes \eta) + O(\delta^2)$$
 does not vanish for $\delta$ sufficiently small.

Finally let us prove that $k(\delta)\in i \R$. As a consequence of Proposition \ref{propinv}.ii,  we have that the function $s_{\delta}$, defined on $\R\cap B(0,\delta)$ by  $s_{\delta}(\alpha)=i( i\alpha+w_{\delta}(i\alpha))$ is real valued. Moreover, using  \eqref{E:expw} for $\delta$ small, $s_{\delta}(0)<0$ and $s_{\delta}(-\delta)>0$. In consequence,  this function  admits a root $\alpha(\delta)$ which is real. By uniqueness, $k(\delta)=i\alpha(\delta)$. 
 \end{proof}

\section{Upper Thresholds}
\label{R:higher}
We now extend our analysis to  the upper   thresholds.  We  will show that if  $\lambda_{q_0}$ is an eigenvalue of multiplicity $m_{0}  \geq  1$  of $(-\Delta_{\omega})$,    then $m_0$ is  a bound for  the  number   of  resonances   around $\lambda_{q_0}$.

 Let $(\psi_{q_0,{j}})_{j=1,\ldots, m_{0}}$  be a normalized basis of $\ker(-\Delta_{\omega}-\lambda_{q_{0}})$. In analogy with \eqref{E:Defmu},  for $1\leq j,l\leq m_0$ define 
\bel{6oct17}\ba{ll}\mu_{j,l}&= \displaystyle{\langle \partial_\varphi\psi_{q_0,j}|\pi_{q_0} \partial_\varphi\psi_{q_0,l}\rangle\,||\varepsilon||^2}\\[1em]
 &+ \frac12 \displaystyle{\sum_{q, \,  \,\lambda_q\neq\lambda_{q_0}}(\lambda_q-\lambda_{q_0})\langle \partial_\varphi\psi_{q_0,j}| \pi_q \partial_\varphi\psi_{q_0,l}\rangle\langle (D^2_3+\lambda_q-\lambda_{q_0})^{-1}\varepsilon|\varepsilon\rangle,}\ea\ee 
 and let  $\Upsilon_{q_0}$ be   the matrix $(\mu_{j,l})$. 

Denote by $r_{0}:=\min(\sqrt{|\lambda_{q_0}-\lambda_{q_{0}-1}|},\sqrt{|\lambda_{q_0+1}-\lambda_{q_{0}}|})$ and  $\C^{++}:= \{k \in \C^{+}; \;  {\rm Re \,}  k >0 \}$. \Bk 


\begin{theorem}
\label{T:higher}
 Suppose that $\lambda_{q_0}$ is an eigenvalue of multiplicity $m_{0}\geq 1$  of $(-\Delta_{\omega})$,   that   $\varepsilon: \R\to \R$ is  a  non-zero $C^1$-function  satisfying \eqref{hypeps} with $\alpha >2r_{0}$, \Bk and that  $\mathcal{D}\subset B(0,r_0)$ is  a compact neighborhood of zero. Then, for all $\delta$ sufficiently small,  the operator-valued function $k\mapsto (H(\delta)-\lambda_{q_{0}}-k^2)^{-1}$, initially defined in  $\C^{++}$,   admits a meromorphic extension  on  $\cD$. This extension has at most $m_{0}$ poles, counted with multiplicity.  These poles are among the zeros $(k_l(\delta))_{1 \leq l \leq m_0}$ of some determinant, which satisfy 
$$k_l(\delta)= -i\nu_{q_0,l}\,\delta^2+ o(\delta^2) , \quad \delta\downarrow 0,$$
where  $(\nu_{q_0,l})_{1 \leq l \leq m_0}$ are the eigenvalues \Bk of the matrix $\Upsilon_{q_0}$. 
\end{theorem}
\begin{proof}
Some points in this proof  are   close to what has been done for the first threshold.  We  will  keep the same notations and explain how to modify the arguments of the previous sections. \Bk 
In analogy with section \ref{S1} set
$$\Phi_{{j},\delta}:= -\frac{i}{2}((\partial_\varphi  \psi_{q_0,j} \otimes \eta^{-1}  \varepsilon' )
+
\delta(\partial^2_\varphi \psi_{q_0,j} \otimes  \eta^{-1} \varepsilon^2 )) \ \ \mbox{and} \ \ 
K_0:= \sum_{j}|\Phi_{j,\delta}\rangle\langle \psi_{q_0,j} \otimes \eta|.
$$
Then, the analog of Lemma \ref{Lholo0} still holds.  Here, since $\lambda_{q_0}$ is in the interior of the essential spectrum, the resolvent $(H(\delta)-z)^{-1}$ is initially defined for ${\rm Im \,} z>0$ near $z=\lambda_{q_0}$ and the extension of the  weighted  resolvent  is done with respect to $k= \sqrt{z-\lambda_{q_0}}$ from  $\C^{++}$ to a neighborhood of $k=0$.\Bk 

Also,  as in the proof of Theorem \ref{T:main}, we have for  $k \in \C^{++}$,  \Bk with $R_{0}(k):=(H_{0}-\lambda_{q_{0}}-k^2)^{-1}$ (and similar notation for $R(k)$): \Bk
\bel{E:KeyId}
\eta R(k) \eta =\eta R_{0}(k)\eta ({\rm Id}+K)^{-1}({\rm Id}+\delta T(\delta,k))^{-1},
\ee
where $$K:=\frac{\delta}{k}({\rm Id}+\delta T(\delta,k))^{-1}K_{0}=\frac{\delta}{k}\sum_{j=1}^{m_{0}}|\tilde{\Phi}_{{j},\delta}\rangle\langle \psi_{q_0,{j}} \otimes \eta|$$
is now of rank $m_{0}$, with obvious notation for $\tilde{\Phi}_{{j},\delta}$.
 
Next,  let  $\Pi_{\delta}^{\perp}$  be  the projection over   $\Big( \ker(-\Delta_{\omega}-\lambda_{q_{0}}) \otimes\, \mathrm{span}\{\eta\} \Big)^\perp$  \Bk in the direction of $\mathrm{span}\{\tilde{\Phi}_{{1},\delta},...,\tilde{\Phi}_{{m_{0}},\delta}\}$
 and $\Pi_{\delta}:={\rm Id}-\Pi_{\delta}^{\perp}$. 
Then, the matrix of $({\rm Id}+K)\Pi_{\delta}$ 
in  the basis $\{\tilde{\Phi}_{j,\delta}\}_{1}^{m_{0}}$ is given, for $k\neq0$, by
\bel{29sep17}\frac{1}{k}\begin{bmatrix}k+\omega_{1,1,\delta}(k)&\dots&\omega_{m_{0},1,\delta}(k)\\
\vdots&\ddots&\vdots\\
\omega_{1,m_{0},\delta}(k)&\dots&k+\omega_{m_{0},m_{0},\delta}(k)
\end{bmatrix}:=\frac{1}{k}M_{\delta}(k)\ee
where we have set
$w_{j,l,\delta}(k)= \delta \langle \tilde{\Phi}_{{j},\delta}| \psi_{q_0,l} \otimes \eta \rangle.
$
Assume that $M_{\delta}(k)$ is invertible, then by \eqref{E:KeyId}
\begin{align*}
\eta R(k)\eta =& \Big(\frac{i}{2k}  \sum_{j}|\psi_{q_0,j} \otimes \eta\rangle\langle \psi_{q_0,j} \otimes \eta|+ A_0(k) \Big)
\Big( \Pi_{\delta}^{\perp} + kM_{\delta}(k)^{-1} \Pi_{\delta}\Big) \; ({\rm Id}+ \delta T(\delta,k) )^{-1}
\\
=&\left( \frac{i}{2}    \sum_{j}|\psi_{q_0,j} \otimes \eta\rangle\langle \psi_{q_0,j} \otimes \eta| M_{\delta}(k)^{-1} \Pi_{\delta} +A_{0}(k)\left(\Pi_{\delta}^{\perp}+k M_{\delta}(k)^{-1}\Pi_{\delta}\right) \right)\; ({\rm Id}+ \delta T(\delta,k) )^{-1}.
\end{align*}
\Bk
%
In consequence, since the $w_{l,k,\delta}$ are holomorpic, $\eta R \eta$ admits a meromorphic extension  to $\cD$,  and the poles of this extension are among the poles of $\Big(\frac{i}{2} \Bk \sum_{j}|\psi_{q_0,j} \otimes \eta\rangle\langle \psi_{q_0,j} \otimes \eta| + k A_{0}(k) \Big)M_{\delta}(k)^{-1} \Pi_{\delta} $.  Evidently,  the poles  are included in the  set of  zeros of the determinant of $M_{\delta}(k)$.

Define
$$\Delta(k,\delta):={\rm det}(M_{\delta}(k)).$$
We can check as in Proposition \ref{propinv}  that 
\bel{1}w_{j,l,\delta}(k)= i\mu_{j,l} \delta^2+O(\delta^3)+\delta^2 k g_{j,l}(k,\delta),\ee
where the $\mu_{j,l}$ are  given   by \eqref{6oct17}. Then 
$$\Delta(k,\delta)= \delta^{2m_0} \Bk {\rm det}(k\delta^{-2}+ i\mu_{j,l}+O(\delta)+ k g_{j,l}(k,\delta)),
$$
and the  zeros of $\Delta(\cdot,\delta)$ are the complex numbers  of the form $k=u\delta^2$, with $u$ being a zero of 
$$
\tilde{\Delta}(u,\delta): =  {\rm det}(u+ i\mu_{j,l}+O(\delta)+ \delta^2 u g_{j,l}(\delta^2u,\delta)).
$$
Since 
\bel{5oct17}\tilde{\Delta}(u,\delta)=\tilde{\Delta}(u,0)+\delta h(u,\delta)={\rm det}(u+ i\mu_{j,l})+\delta h(u,\delta),\ee where $h$ is an analytic  function in $u$ and $\delta$, taking the  ball $B(0,C)$ 
 with $C$ larger than the modulus of the larger eigenvalue of $\Upsilon_{q_0}$ and applying Rouche theorem,
 we conclude that all the zeros of $\tilde{\Delta}( \cdot \Bk ,\delta)$ are inside this ball for $\delta$ sufficiently small.  \Bk Moreover, 
if we denote by $\nu_{q_0,l}$ the eigenvalues of $\Upsilon_{q_0}$,  \eqref{5oct17} yields 
$u_{q_0,l}(\delta)= -i \Bk  (\nu_{q_0,l}+o(1)).$ This  immediately implies that all the zeros of $\Delta( \cdot \Bk,\delta)$ in $\cD$,   denoted by $k_l$, are inside the ball $B(0,C\delta^2),$ and satisfy
$$k_{l}(\delta)=-  i \delta^2 \Bk (\nu_{q_0,l}+o(1)).$$
  \Bk 
%
%
%
\end{proof}
\begin{remark}\label{4oct17}
 In the last theorem, if  $m_0=1$, we are able to  obtain extra  information. For instance, as in Theorem \ref{T:main}, for the unique  zero of the determinant,  $k_{1}(\delta)$, we have that 
$k_{1}(\delta)= -i \mu_{q_0}\delta^2+O(\delta^3)$
with 
$$
\mu_{q_0}:=\mu_{1,1}= \tfrac{1}{2}\sum_{q\neq q_{0}}(\lambda_{q}-\lambda_{q_{0}})\langle \partial_{\varphi}\psi_{q_0}|\pi_{q}\partial_{\varphi}\psi_{q_{0}}\rangle\langle \varepsilon|(D_{3}^2+\lambda_{q}-\lambda_{q_{0}})^{-1}\varepsilon\rangle.$$\Bk Then, as in the proof of Theorem \ref{T:main}, $k_{q_{0}}$ is a pole of rank one when $\mu_{q_{0}}\neq 0$.  
It is also  important to notice that, for $q<q_{0}$, the operator $(D_{3}^2+\lambda_{q}-\lambda_{q_{0}})^{-1}$ has to be understood as the limit of $(D_{3}^2+\lambda_{q}-\lambda_{q_{0}}-k^2)^{-1}$, acting in weighted spaces, when $k\to 0$. It is not a selfadjoint operator anymore, therefore $\mu_{q_0}$ is not necessarily real. Actually, in general, it has a non zero imaginary part coming from the first terms when   $q<q_0$.  Indeed,  thanks to \eqref{3sep17a}, for $q<q_0$, the imaginary part of $ 2  (\lambda_{q_{0}}-\lambda_{q})^{-1/2} \Bk \, \langle \varepsilon|(D_{3}^2+\lambda_{q}-\lambda_{q_{0}})^{-1}\varepsilon\rangle$ is given by:
$$
-\Big(\int_\R \cos( \sqrt{\lambda_{q_0}-\lambda_q} x_3) \, \varepsilon (x_3) \, d{x_3} \Big)^2 - \Big(\int_\R \sin( \sqrt{\lambda_{q_0}-\lambda_q} x_3) \, \varepsilon (x_3) \, d{x_3} \Big)^2 = - \sqrt{2\pi} | \widehat{\varepsilon} ( \sqrt{\lambda_{q_0}-\lambda_q}) |^2.
$$
where $\widehat{\varepsilon}$ if the Fourier transform of $\varepsilon$.  Then, the  imaginary part of $\mu_{q_0}$ is:
$${\hbox{Im}}(\mu_{q_0})= - \tfrac{\sqrt{2\pi(\lambda_{q_0}-\lambda_q)}}{ 4}\sum_{q < q_{0}}  \|\pi_{q}\partial_{\varphi}\psi_{q_{0}}\|^2  \, | \widehat{\varepsilon} ( \sqrt{\lambda_{q_0}-\lambda_q}) |^2.
$$
 This identity allows to give sufficient conditions on the eigenfunctions of $-\Delta_{\omega}$ and on $\hat{\varepsilon}$, so that $\mu_{q_{0}}\neq0$,   giving rise to a unique resonance  of multiplicity one, with ${\rm Re\,} k_1(\delta)<0$. 
\end{remark}
\bibliographystyle{plain}
\bibliography{biblio,bibliopof}
\end{document}